\documentclass[journal]{IEEEtran}
\usepackage{amsmath}
\usepackage{amssymb}
\usepackage{amsfonts}
\usepackage{graphicx}
\usepackage{epsfig}
\usepackage{subfigure}
\usepackage{psfrag}
\usepackage{color}
\usepackage{cases}

\begin{document}

\title{How to Diagonalize a MIMO Channel with Arbitrary Transmit Covariance?}

\author{Liang Liu, \IEEEmembership{Member,~IEEE}, and Rui Zhang, \IEEEmembership{Senior Member,~IEEE}
\thanks{Manuscript received February 5, 2016, revised March 26, 2016, accepted April 22, 2016. The associate editor
coordinating the review of this letter and approving it for
publication was Dr. Saif Khan Mohammed.}
\thanks{L. Liu is with the
Department of Electrical and Computer Engineering, University of Toronto (e-mail:lianguot.liu@utoronto.ca).}
\thanks{R. Zhang is with the
Department of Electrical and Computer Engineering, National
University of Singapore (e-mail:elezhang@nus.edu.sg).}}

\maketitle

\begin{abstract}
Multiple-input multiple-output (MIMO) or multi-antenna communication is a key technique to achieve high spectral efficiency in wireless systems. For the point-to-point MIMO channel, it is a well-known result that the channel singular value decomposition (SVD) based linear precoding and decoding achieves the channel capacity, which also diagonalizes the MIMO channel into parallel single-input single-output (SISO) sub-channels for independent encoding and decoding. However, in multi-user MIMO systems, the optimal transmit covariance of each MIMO link is generally not its channel SVD based as a result of the control and balance of the co-channel interference among users. Thus, it remains unknown whether the linear precoding/decoding strategy is still able to achieve the capacity of each MIMO link and yet diagonalize its MIMO channel, with a given set of optimal transmit covariance of all users. This letter solves this open problem by providing a closed-form capacity-achieving linear precoder/decoder design that diagonalizes a MIMO channel with arbitrary transmit covariance. Numerical examples are also provided to validate the proposed solution in various multi-user MIMO systems.
\end{abstract}




\setlength{\baselineskip}{1.0\baselineskip}
\newtheorem{definition}{\underline{Definition}}[section]
\newtheorem{fact}{Fact}
\newtheorem{assumption}{Assumption}
\newtheorem{theorem}{\underline{Theorem}}[section]
\newtheorem{lemma}{\underline{Lemma}}[section]
\newtheorem{corollary}{\underline{Corollary}}[section]
\newtheorem{proposition}{\underline{Proposition}}[section]
\newtheorem{example}{\underline{Example}}[section]
\newtheorem{remark}{\underline{Remark}}[section]
\newtheorem{algorithm}{\underline{Algorithm}}[section]
\newcommand{\mv}[1]{\mbox{\boldmath{$ #1 $}}}

\vspace{-15pt}
\section{Introduction}\label{sec:Introduction}\vspace{-5pt}
Multi-antenna or so-called multiple-input multiple-output (MIMO) technique has received consistently significant attention in both single-user and multi-user wireless communications to achieve enormous spatial multiplexing and/or diversity gains (see, e.g., \cite{Telarta99} -- \cite{Marzetta10}). When the channel is perfectly known at both the transmitter and receiver and under the point-to-point single-user setup, it is a well-known result that the MIMO channel singular value decomposition (SVD) based linear precoding and decoding achieves the capacity \cite{Cover}. Moreover, the capacity-achieving SVD-based linear precoder and decoder diagonalizes the MIMO channel into parallel single-input single-output (SISO) sub-channels for independent encoding and decoding, which makes spatial multiplexing practically implementable with low transceiver complexity. In \cite{Jiang05}, a uniform channel decomposition (UCD) based linear precoding is proposed to decompose the MIMO channel into parallel sub-channels with equal single-to-noise ratio (SNR) to apply the same modulation scheme in practice. This linear precoder design also achieves the MIMO channel capacity; however, unlike the SVD-based design, the receiver of UCD needs to apply the non-linear minimum mean-squared-error (MMSE) based decoding with successive interference cancellation (SIC) \cite{Jiang05}. On the other hand, for the case when the channel is not known at the transmitter, the Vertical Bell Labs Layered Space-Time (V-BLAST) architecture is proposed in \cite{Foschini96}, which applies the isotropic transmission without precoding and the non-linear MMSE receiver with SIC. Interestingly, it is shown in \cite{Cioffi01} that the V-BLAST receiver is a special case of the generalized decision feedback equalizer (GDFE) for generic MIMO systems. However, like the UCD-based transmission, the non-linear (instead of linear) receiver is necessary to achieve the MIMO channel capacity with V-BLAST.

This letter extends the study of linear transceiver design for achieving the MIMO channel capacity with perfect channel knowledge by addressing the following question: given an arbitrary transmit covariance in a point-to-point MIMO channel, is there always a linear precoder and decoder solution that not only achieves the channel capacity but also diagonalizes the MIMO channel for parallel SISO processing? This is mainly motivated by multi-user MIMO communication systems with the co-channel interference among the users. In these systems, the optimal transmit covariance of each user's MIMO link for achieving the system's maximum throughput such as weighted sum-rate of all users depends on the user's direct MIMO channel as well as all other users' direct and cross-link MIMO channels (see, e.g., \cite{Luo11} -- \cite{Zhang12}), which is thus not its direct MIMO channel SVD based in general and cannot diagonalize the direct MIMO channel. This letter solves this problem by providing a closed-form capacity-achieving linear precoder/decoder design that diagonalizes a MIMO channel with arbitrary transmit covariance. Moreover, rich numerical examples are provided to validate the proposed solution in various multi-user MIMO systems.

{\it Notation}: $\mv{I}$ and $\mv{0}$  denote an
identity matrix and an all-zero matrix, respectively, with
appropriate dimensions. For a square matrix $\mv{S}$, $\mv{S}^{-1}$ and ${\rm det}(\mv{S})$
denote its inverse (if $\mv{S}$ is full-rank) and determinant, respectively; $\mv{S}\succeq\mv{0}$ means that $\mv{S}$ is positive semi-definite. For a matrix
$\mv{M}$ of arbitrary size, $\mv{M}^{H}$ and $\mv{M}^{T}$ denote the
conjugate transpose and transpose of $\mv{M}$, respectively; ${\rm rank}(\mv{M})$ denotes the rank of $\mv{M}$. ${\rm
diag}(x_1,\cdots,x_K)$ denotes a diagonal matrix
with diagonal elements given by $x_1,\cdots,x_K$. The
distribution of a circularly
symmetric complex Gaussian (CSCG) random vector with mean $\mv{x}$ and
covariance matrix $\mv{\Sigma}$ is denoted by
$\mathcal{CN}(\mv{x},\mv{\Sigma})$; and $\sim$ stands for
``distributed as''. $\mathbb{C}^{x \times y}$ denotes the space of
$x\times y$ complex matrices. $\|\mv{x}\|$ denotes the Euclidean norm of a complex vector
$\mv{x}$. $\min(a,b)$ denotes the minimum of two real numbers $a$ and $b$. $E[\cdot]$ denotes the statistical expectation.

\vspace{-15pt}
\section{System Model}\label{sec:System Model}\vspace{-5pt}
Consider a point-to-point MIMO channel consisting of one transmitter equipped with $M> 1$ antennas and one receiver with $N> 1$ antennas. We assume the MIMO channel is perfectly known at both the transmitter and receiver. The baseband transmitted signal is given as
\begin{align}
\mv{x}=\mv{V}\mv{s},
\end{align}where $\mv{s}=[s_1,\cdots,s_D]^T\sim \mathcal{CN}(\mv{0},\mv{I})$ denotes the information-bearing signals via spatial multiplexing over $D\leq \min(M,N)$ independent data streams, and $\mv{V}\in \mathbb{C}^{M\times D}$ denotes the linear precoder applied at the transmitter. The baseband received signal is then given as
\begin{align}\label{eqn:single-user signal}
\tilde{\mv{y}}=\tilde{\mv{H}}\mv{x}+\tilde{\mv{z}}=\tilde{\mv{H}}\mv{V}\mv{s}+\tilde{\mv{z}},
\end{align}where $\tilde{\mv{H}}\in \mathbb{C}^{N\times M}$ denotes the MIMO channel, and $\tilde{\mv{z}}=[\tilde{z}_1,\cdots,\tilde{z}_N]^T\sim \mathcal{CN}(\mv{0},\mv{S}_z)$ denotes the noise at the multi-antenna receiver with the covariance matrix $\mv{S}_z\triangleq E[\tilde{\mv{z}}\tilde{\mv{z}}^H]$. Without loss of generality, we assume $D\leq {\rm rank}(\tilde{\mv{H}})$. Note that $\tilde{\mv{z}}$  is in general not spatially white since it may include the co-channel interference (assumed to be Gaussian distributed) from other transmitters in multi-user communication setups (see Section \ref{sec:Numerical Results} for examples). At the receiver, without loss of optimality, a noise-whitening filter can be applied to obtain
\begin{align}\label{eqn:channel}
\mv{y}=\mv{S}_z^{-\frac{1}{2}}\tilde{\mv{y}}=\mv{H}\mv{V}\mv{s}+\mv{z},
\end{align}where $\mv{H}=\mv{S}_z^{-\frac{1}{2}}\tilde{\mv{H}}\in \mathbb{C}^{N\times M}$ is the effective MIMO channel and $\mv{z}=\mv{S}_z^{-\frac{1}{2}}\tilde{\mv{z}}\in \mathbb{C}^{N\times N}$ with $\mv{z}\sim \mathcal{CN}(\mv{0},\mv{I})$ denotes the effective Gaussian noise.

First, we consider the case of linear receiver. In this case, the received signal in (\ref{eqn:channel}) is multiplied by a linear decoding matrix $\mv{U}^H\in \mathbb{C}^{D\times N}$, i.e., \begin{align}\label{eqn:receive beamforming}
\hat{\mv{y}}=\mv{U}^H\mv{y}=\mv{U}^H\mv{H}\mv{V}\mv{s}+\mv{U}^H\mv{z}.
\end{align}Let $\hat{\mv{y}}\triangleq[\hat{y}_1,\cdots,\hat{y}_D]^T$. A pair of linear precoder and linear decoder is called ``diagonalizing'' the MIMO channel $\mv{H}$ if $\mv{U}^H\mv{H}\mv{V}$ in (\ref{eqn:receive beamforming}) is a diagonal matrix. As a result, the MIMO channel in (\ref{eqn:receive beamforming}) is  decomposed into $D$ non-interfering parallel SISO sub-channels given by
\begin{align}\label{eqn:channel diagonalization}
\hat{y}_d=\mv{u}_d^H\mv{H}\mv{v}_ds_d+\hat{z}_d, ~~~ d=1,\cdots,D,
\end{align}where $\mv{u}_d$ and  $\mv{v}_d$ denote the $d$th columns of $\mv{U}$ and $\mv{V}$, respectively, and $\hat{z}_d=\mv{u}_d^H\mv{z}\sim \mathcal{CN}(0,\|\mv{u}_d\|^2)$. The SNR for decoding the information in $s_d$ is thus given by
\begin{align}
\gamma_d=\frac{| \mv{u}_d^H\mv{H}\mv{v}_d|^2}{\|\mv{u}_d\|^2}, ~~~ d=1,\cdots,D.
\end{align}As a result, the achievable sum-rate in bits/sec/Hz (bps/Hz) over all $D$ sub-channels with a given pair of linear precoder $\mv{V}$ and linear decoder $\mv{U}^H$ is
\begin{align}\label{eqn:achievable rate}
R=\sum\limits_{d=1}^D\log_2(1+\gamma_d)=\sum\limits_{d=1}^D\log_2\left(1+\frac{| \mv{u}_d^H\mv{H}\mv{v}_d|^2}{\|\mv{u}_d\|^2}\right).
\end{align}An illustration of the above linear precoding/decoding scheme that diagonalizes the MIMO channel is given in Fig. \ref{fig1}.

\begin{figure}
\begin{center}
\scalebox{0.45}{\includegraphics*[angle=90]{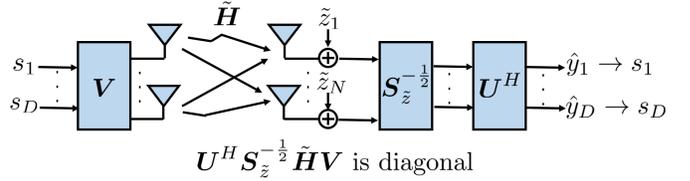}}
\end{center}\vspace{-10pt}
\caption{Illustration of the linear transceiver architecture that diagonalizes MIMO channels.}\label{fig1} \vspace{-15pt}
\end{figure}

Next, we consider the capacity of the MIMO channel where the covariance matrix of the transmit signal is constrained to be a given $\mv{S}_x\triangleq E[\mv{x}\mv{x}^H]\succeq \mv{0}$. Given $\mv{S}_x$, the capacity of the point-to-point MIMO channel given in (\ref{eqn:channel}) in bps/Hz is \cite{Cover}
\begin{align}\label{single-user capacity}
C(\mv{S}_x)=\log_2 \det \left(\mv{I}+\mv{H}\mv{S}_x\mv{H}^H\right).
\end{align}In general, for an arbitrary transmit covariance $\mv{S}_x\succeq \mv{0}$, the corresponding linear precoder $\mv{V}$ obtained directly via the eigenvalue value decomposition (EVD) of $\mv{S}_x$, i.e., $\mv{S}_x=\mv{V}\mv{V}^H$ (with $\mv{V}$ given in (\ref{eqn:zf precoder})), and a linear decoder $\mv{U}^H$ which jointly diagonalize the MIMO channel $\mv{H}$ is capacity-suboptimal, i.e., $R\leq C(\mv{S}_x)$. However, there is a special case when $R=C(\mv{S}_x)$ if the linear precoder $\mv{V}$ is designed based on the SVD of the MIMO channel $\mv{H}$. Specifically, let $D_H={\rm rank}(\mv{H})$ and express the truncated SVD of $\mv{H}$ as $\mv{H}=\mv{U}_H\mv{\Lambda}_H\mv{V}_H^H$, where $\mv{U}_H\in \mathbb{C}^{N\times D_H}$ with $\mv{U}_H^H\mv{U}_H=\mv{I}$, $\mv{V}_H\in \mathbb{C}^{M\times D_H}$ with $\mv{V}_H^H\mv{V}_H=\mv{I}$, and $\mv{\Lambda}_H$ is a $D_H$-by-$D_H$ positive diagonal matrix. Let\begin{align}
& \mv{V}=\mv{V}_H\mv{P}^{\frac{1}{2}}, \label{optimal transmit signal} \\
& \mv{U}^H=\mv{U}_H^H, \label{single-user signal eigenmode}
\end{align}where $\mv{P}$ is a $D_H$-by-$D_H$ positive diagonal matrix. Then, it can be easily verified that the pair of linear precoder $\mv{V}$ and linear decoder $\mv{U}^H$ given in (\ref{optimal transmit signal}) and (\ref{single-user signal eigenmode}) diagonalize the MIMO channel $\mv{H}$, which is referred to as the channel SVD based linear precoder/decoder design. Furthermore, if the water-filling power allocation is applied to design $\mv{P}$, it is known that the resulting transmit covariance,
\begin{align}\label{optimal covariance}
\mv{S}_x=\mv{V}_H\mv{P}\mv{V}_H^H,
\end{align}achieves the capacity of the MIMO channel, i.e., $R=C(\mv{S}_x)$ \cite{Cover}.

However, if the transmit covariance $\mv{S}_x$ is not in the form of (\ref{optimal covariance}), the channel SVD based linear precoding/decoding cannot be applied, as shown in the following example. Define the EVD of $\mv{S}_x$ as $\mv{S}_x=\mv{U}_x\mv{\Lambda}_x\mv{U}_x^H$, where $\mv{U}_x\in \mathbb{C}^{M\times D}$ with $\mv{U}_x^H\mv{U}_x=\mv{I}$, and $\mv{\Lambda}_x$ is a $D$-by-$D$ positive diagonal matrix. Then, given any $\mv{S}_x\succeq \mv{0}$, consider the linear precoder designed based on the EVD of $\mv{S}_x$ as
\begin{align}\label{eqn:zf precoder}
\mv{V}=\mv{S}_x^{\frac{1}{2}}=\mv{U}_x\mv{\Lambda}_x^{\frac{1}{2}}.
\end{align}In order to diagonalize the MIMO channel, the receiver applies a linear decoder $\mv{U}^H$ such that $\mv{U}^H\mv{H}\mv{V}$ is diagonal. One possible choice of $\mv{U}^H$ is the zero-forcing (ZF) receiver:
\begin{align}\label{eqn:zf decoder}
\mv{U}^H=((\mv{H}\mv{V})^H\mv{H}\mv{V})^{-1}(\mv{H}\mv{V})^H.
\end{align}Next, we present a numerical example to compare the achievable rate with the above linear precoder/decoder to the channel capacity. We consider $M=N=2$, and
\begin{align*}
 \mv{H}=\left[\begin{array}{rr}0.8147 & \hspace{-5pt} 0.1270 \\ 0.9058 & \hspace{-5pt} 0.9134 \end{array} \right],  ~  \mv{S}_x=\left[\begin{array}{rr}0.2896 & \hspace{-5pt}-0.5654 \\ -0.5654 & \hspace{-5pt}1.8275 \end{array} \right].
\end{align*}With the linear precoder given in (\ref{eqn:zf precoder}) and the linear ZF receiver given in (\ref{eqn:zf decoder}) for this MIMO channel, we have
\begin{align*}
\mv{V}=\left[\begin{array}{rr}-0.4423 &\hspace{-5pt} 0.3066 \\ 1.3481 & \hspace{-5pt}0.1006 \end{array} \right], ~ \mv{U}=\left[\begin{array}{rr}-1.2832 & \hspace{-5pt}2.8846 \\ \hspace{-5pt}0.9116 & 0.6576 \end{array} \right].
\end{align*}It can be shown that with the above channel-diagonalizing linear precoder and decoder, the achievable rate given in (7) is $R=0.6452$ bps/Hz, while the capacity of this channel with the given $\mv{S}_x$ can be computed from (\ref{single-user capacity}) as $C(\mv{S}_x)=1.0103$ bps/Hz. Evidently, we have $R<C(\mv{S}_x)$ in this example. Notice that with the above precoder, the channel capacity $C(\mv{S}_x)$ needs to be achieved with non-linear MMSE decoder with SIC \cite{Cioffi01}.

\vspace{-10pt}
\section{Problem Formulation}\label{sec:Problem Formulation}\vspace{-2pt}
In this letter, for an arbitrary transmit covariance matrix $\mv{S}_x\succeq \mv{0}$, we aim to find a pair of linear precoder $\mv{V}$ and decoder $\mv{U}^H$ that diagonalizes the MIMO channel and yet achieves the channel capacity with the given $\mv{S}_x$, i.e., they need to satisfy the following three conditions:
\begin{align}
& \sum\limits_{d=1}^D\log_2\left(1+\frac{| \mv{u}_d^H\mv{H}\mv{v}_d|^2}{\|\mv{u}_d\|^2}\right)=\log_2 \det \left(\mv{I}+\mv{H}\mv{S}_x\mv{H}^H\right), \label{eqn:condition 1} \\
& \mv{U}^H\mv{H}\mv{V}={\rm diag}(\mv{u}_1^H\mv{H}\mv{v}_1,\cdots,\mv{u}_D^H\mv{H}\mv{v}_D), \label{eqn:condition 2} \\
& \mv{V}\mv{V}^H=\mv{S}_x. \label{eqn:condition 3}
\end{align}In the above, (\ref{eqn:condition 1}) is the capacity-achieving condition that ensures $R=C(\mv{S}_x)$, with $R$ and $C(\mv{S}_x)$ given in (\ref{eqn:achievable rate}) and (\ref{single-user capacity}), respectively; (\ref{eqn:condition 2}) is the channel diagonalization condition; while (\ref{eqn:condition 3}) is the transmit covariance condition.

\vspace{-10pt}
\section{Optimal Solution}\label{sec:Optimal Solution}\vspace{-2pt}
In this section, we show that given any transmit covariance $\mv{S}_x \succeq \mv{0}$, there always exists a linear precoder/decoder design that can diagonalize the MIMO channel and also achieve the channel capacity. Specifically, the following theorem presents a closed-form precoder/decoder solution that satisfies conditions in (\ref{eqn:condition 1}) -- (\ref{eqn:condition 3}) simultaneously.
\begin{theorem}\label{theorem1}
Given any transmit covariance $\mv{S}_x\succeq \mv{0}$ and $\mv{S}_x^{\frac{1}{2}}$ (see (\ref{eqn:zf precoder})), with $D={\rm rank}(\mv{S}_x^{\frac{1}{2}})={\rm rank}(\mv{H}\mv{S}_x^{\frac{1}{2}})$,\footnote{This holds without loss of generality since if ${\rm rank}(\mv{S}_x)>{\rm rank}(\mv{H}\mv{S}_x^{\frac{1}{2}})$, we can always construct a new transmit covariance $\bar{\mv{S}}_x$ with ${\rm rank}(\bar{\mv{S}}_x) = {\rm rank}(\mv{H}\bar{\mv{S}}_x^{\frac{1}{2}})$ which achieves the same capacity of the MIMO channel $\mv{H}$ with the given $\mv{S}_x$.} let the truncated SVD of $\mv{\Phi}=\mv{H}\mv{S}_x^{\frac{1}{2}}$ be given by
\begin{align}\label{eqn:Phi}
\mv{\Phi}=\mv{H}\mv{S}_x^{\frac{1}{2}}=\mv{U}_{\Phi}\mv{\Lambda}_{\Phi}\mv{V}_{\Phi}^H,
\end{align}where $\mv{U}_{\Phi}\in \mathbb{C}^{N\times D}$ with $\mv{U}_{\Phi}^H\mv{U}_{\Phi}=\mv{I}$, $\mv{V}_{\Phi}\in \mathbb{C}^{D\times D}$ with $\mv{V}_{\Phi}^H\mv{V}_{\Phi}=\mv{V}_{\Phi}\mv{V}_{\Phi}^H=\mv{I}$, and $\mv{\Lambda}_{\Phi}={\rm diag}(\phi_1, \cdots , \phi_D)$, with $\phi_d>0$, $d=1,\cdots,D$. Then, the following linear precoder/decoder design satisfies conditions (\ref{eqn:condition 1}) -- (\ref{eqn:condition 3}):
\begin{align}
& \mv{V}=\mv{S}_x^{\frac{1}{2}}\mv{V}_{\Phi},  \label{eqn:feasible V} \\
& \mv{U}^H=\mv{U}_{\Phi}^H.  \label{eqn:feasible U}
\end{align}
\end{theorem}

\begin{proof}
For convenience, we first verify the channel diagonalization condition given in (\ref{eqn:condition 2}). With the linear precoder and decoder given in (\ref{eqn:feasible V}) and (\ref{eqn:feasible U}), respectively, from (\ref{eqn:Phi}) we have
\begin{align}\label{eqn:channel diagonalization}
\mv{U}^H\mv{H}\mv{V} = \mv{U}_{\Phi}^H\mv{H}\mv{S}_x^{\frac{1}{2}}\mv{V}_{\Phi} = \mv{\Lambda}_{\Phi}={\rm diag}(\phi_1,\cdots,\phi_D).
\end{align}

Next, consider the capacity-achieving condition given in (\ref{eqn:condition 1}). Since for $\mv{U}$ given in (\ref{eqn:feasible U}), we have $\|\mv{u}_d\|=1$, $\forall d$, and from (\ref{eqn:channel diagonalization}), we have  $\mv{u}_d^H\mv{H}\mv{v}_d=\phi_d>0$, $\forall d$, it follows that
\begin{align}
R=& \sum\limits_{d=1}^D\log_2\left(1+\frac{| \mv{u}_d^H\mv{H}\mv{v}_d|^2}{\|\mv{u}_d\|^2}\right)= \sum\limits_{d=1}^D \log_2(1+\phi_d^2) \nonumber \\ \overset{(a)}{=} &\log_2\det(\mv{I}+\mv{\Lambda}_{\Phi}^2) \nonumber \\ \overset{(b)}{=} & \log_2\det(\mv{I}+\mv{U}_{\Phi}\mv{\Lambda}_{\Phi}\mv{V}_{\Phi}^H\mv{V}_{\Phi}\mv{\Lambda}_{\Phi}\mv{U}_{\Phi}^H) \nonumber \\ \overset{(c)}{=} & \log_2\det(\mv{I}+\mv{H}\mv{S}_x\mv{H}^H)=C(\mv{S}_x),
\end{align}where $(a)$ is due to $\mv{\Lambda}_{\Phi}={\rm diag}(\phi_1,\cdots,\phi_D)$, $(b)$ is due to $\mv{V}_{\Phi}^H\mv{V}_{\Phi}=\mv{U}_{\Phi}^H\mv{U}_{\Phi}=\mv{I}$ and the fact that $\log_2 \det(\mv{I}+\mv{A}\mv{B})=\log_2 \det(\mv{I}+\mv{B}\mv{A})$, and $(c)$ is due to (\ref{eqn:Phi}).

Last, consider the transmit covariance condition given in (\ref{eqn:condition 3}). With $\mv{V}$ given in (\ref{eqn:feasible V}) and $\mv{V}_{\Phi}\mv{V}_{\Phi}^H=\mv{I}$, it follows that
\begin{align}
\mv{V}\mv{V}^H=\mv{S}_x^{\frac{1}{2}}\mv{V}_{\Phi}\mv{V}_{\Phi}^H(\mv{S}_x^{\frac{1}{2}})^H=\mv{S}_x.
\end{align}

Theorem \ref{theorem1} is thus proved.
\end{proof}

\begin{remark}\label{remark1}
It is worth noting that if the transmit covariance is given as (\ref{optimal covariance}), then we have $\mv{U}_{\Phi}=\mv{U}_H$ and $\mv{V}_{\Phi}=\mv{I}$ such that $\mv{V}=\mv{S}_x^{\frac{1}{2}}=\mv{V}_H\mv{P}^{\frac{1}{2}}$. As a result, the linear precoding and decoding solution given in (\ref{eqn:feasible V}) and (\ref{eqn:feasible U}) becomes that based on the channel SVD as given in (\ref{optimal transmit signal}) and (\ref{single-user signal eigenmode}).
\end{remark}

Theorem \ref{theorem1} is of practical significance in multi-user MIMO channels, when the optimal transmit covariance of each user is not given as (\ref{optimal covariance}). Specifically, we can first derive the optimal transmit covariance solutions for all users, and then with the obtained optimal transmit covariance of each user apply (\ref{eqn:feasible V}) and (\ref{eqn:feasible U}) to find the corresponding optimal linear precoder/decoder that achieves the channel capacity and also diagonalizes the MIMO channel of each user. Consider the same example given at the end of Section \ref{sec:System Model}. From (\ref{eqn:feasible V}) and (\ref{eqn:feasible U}), we obtain a new pair of channel-diagonalizing linear precoder and decoder as
\begin{align*}
 \mv{V}=\left[\begin{array}{rr}-0.2925 &\hspace{-5pt} 0.4517 \\ 1.2851 & \hspace{-5pt}-0.4195 \end{array} \right], ~ \mv{U}=\left[\begin{array}{rr}-0.0824 & \hspace{-5pt}0.9966 \\ \hspace{-5pt}0.9966 & 0.0825 \end{array} \right],
\end{align*}and the achievable rate given in (\ref{eqn:achievable rate}) is obtained as $R=1.0103$ bps/Hz. Thus, we have $R=C(\mv{S}_x)$ for the given MIMO channel and transmit covariance.

\vspace{-10pt}
\section{Numerical Results}\label{sec:Numerical Results}\vspace{-2pt}
In this section, we provide numerical examples to verify the effectiveness of the proposed design under different multi-user MIMO setups. Due to space limitations, we present the results only for the MIMO interference channel (IC) \cite{Luo11} and MIMO cognitive radio (CR) channel \cite{Rui08}, while the verification also holds for other multi-user MIMO systems, such as the MIMO multiple-access channel (MAC) \cite{Yu04} and  MIMO broadcast channel (BC) \cite{Zhang12}. In the following, we consider two-user systems where each user is equipped with two antennas, and the transmit power constraint for each user is $1$ dB. Moreover, we assume that the covariance of the background noise at each receiver is $\mv{I}$. We consider the real-valued channels for the ease of illustration.

{\bf Example 1: MIMO IC.} In this example, we consider a two-user MIMO IC. Let $\tilde{\mv{H}}_{j,k}$ denote the channel from transmitter $k$ to receiver $j$, $j,k=1,2$, with the following realization:
\vspace{-5pt}\begin{align*}
\begin{array}{ll}  \hspace{-5pt}\tilde{\mv{H}}_{1,1}\hspace{-3pt}=\hspace{-3pt}\left[\begin{array}{rr}\hspace{-3pt}2.0108 & \hspace{-5pt}0.3083 \\ \hspace{-3pt}0.0256 & \hspace{-5pt}-0.9383 \end{array} \hspace{-3pt}\right],
& \hspace{-5pt}\tilde{\mv{H}}_{2,1}\hspace{-3pt}=\hspace{-3pt}\left[\begin{array}{rr}\hspace{-3pt}0.4270  & \hspace{-5pt}\hspace{-3pt}-0.5780 \\ 0.1946 & \hspace{-5pt}0.0199 \end{array} \hspace{-3pt}\right], \\
\hspace{-5pt} \tilde{\mv{H}}_{1,2} \hspace{-3pt}=\hspace{-3pt}\left[\begin{array}{rr}\hspace{-3pt}-0.2253 & \hspace{-5pt}-0.1253 \\ \hspace{-3pt}0.0546 & \hspace{-5pt}-0.0950 \end{array} \hspace{-3pt}\right],
& \hspace{-5pt}\tilde{\mv{H}}_{2,2}\hspace{-3pt}=\hspace{-3pt}\left[\begin{array}{rr}\hspace{-3pt}1.6742 & \hspace{-5pt}0.5301 \\\hspace{-3pt} 0.1250 & \hspace{-5pt}-0.9521 \end{array} \hspace{-3pt}\right].  \label{eqn:1} \end{array} \vspace{-10pt}
\end{align*}For the above MIMO IC, we consider the problem of maximizing the two users' sum-rate as considered in \cite{Luo11}, by treating the interference as additive noise at each receiver. In general, this problem is non-convex and thus difficult to solve optimally. Thus, we apply the weighted sum mean-squared-error minimization (WMMSE) algorithm proposed in \cite{Luo11} to obtain a pair of suboptimal linear precoders for user $1$ and user $2$:
\vspace{-5pt}\begin{align*}
\tilde{\mv{V}}_1=\left[\begin{array}{rr}2.4376 & \hspace{-5pt}-0.6131 \\ 1.4874  & \hspace{-5pt}1.2125 \end{array} \right],
~ \tilde{\mv{V}}_2=\left[\begin{array}{rr}1.9083 & \hspace{-5pt}-1.0758 \\ 1.0682 & \hspace{-5pt}2.0150 \end{array} \right].\vspace{-10pt}
\end{align*}The channel capacities of users 1 and 2 with the corresponding transmit covariance are $6.0141$ bps/Hz and $5.3520$ bps/Hz, respectively. However, it can be shown that with the ZF receiver given in (\ref{eqn:zf decoder}) to diagonalize the effective channel $\mv{H}_{k,k}=\mv{S}_{z_k}^{-\frac{1}{2}}\tilde{\mv{H}}_{k,k}$, where $\mv{S}_{z_k}=\tilde{\mv{H}}_{k,j}\tilde{\mv{V}}_j\tilde{\mv{V}}_j^H\tilde{\mv{H}}_{k,j}^H+\mv{I}$ with $j\neq k$ denotes the covariance matrix of the interference plus noise at receiver $k$, $k=1,2$, the achievable rate of each user is strictly less than its capacity. Instead, with the above precoders, to achieve each user's capacity, the non-linear MMSE receiver with SIC needs to be applied \cite{Cioffi01}.

With Theorem \ref{theorem1}, we can construct the following linear precoder and decoder for user $1$:
\vspace{-5pt}\begin{align*}
\mv{V}_1=\left[\begin{array}{rr}-2.4841 & \hspace{-5pt}0.3836 \\ -1.3681  & \hspace{-5pt}-1.3456 \end{array} \right], ~ \mv{U}_1=\left[\begin{array}{rr}-0.8449 & \hspace{-5pt}0.2119 \\ 0.2368  & \hspace{-5pt}0.9408 \end{array} \right].\vspace{-10pt}
\end{align*}With this new design, it can be shown that $\mv{V}_1\mv{V}_1^H=\tilde{\mv{V}}_1\tilde{\mv{V}}_1^H$, i.e., with the same transmit covariance. Moreover, $\mv{U}_1^H\mv{H}_{1,1}\mv{V}_1$ is a diagonal matrix, i.e., the MIMO channel is diagonalized. Then, by decoding each data stream independently over the two parallel sub-channels, the achievable sum-rate of user 1 is $6.0141$ bps/Hz, which is the same as its channel capacity. Similar result also applies to user $2$, the details of which are omitted for brevity.

{\bf Example 2: MIMO CR.} Next, we consider a MIMO CR network consisting of one secondary user and one primary user under spectrum sharing. The channels from the secondary user transmitter to its receiver and the primary user receiver are denoted by $\tilde{\mv{H}}$ and $\tilde{\mv{G}}$, respectively. For convenience, we set $\tilde{\mv{H}}=\tilde{\mv{H}}_{1,1}$ and $\tilde{\mv{G}}=\tilde{\mv{H}}_{2,1}$, where $\tilde{\mv{H}}_{1,1}$ and $\tilde{\mv{H}}_{2,1}$ are given in Example 1. We assume there is no interference from the primary transmitter to the secondary receiver. With this channel setup, we maximize the capacity of the secondary link subject to the interference temperature (IT) constraint at the primary receiver as considered in \cite{Rui08}. The IT constraint is set such that the total received power from the secondary transmitter at the two antennas of the primary receiver needs to be no larger than 2. This problem is convex, and thus can be efficiently solved by CVX \cite{Boyd11}. The optimal transmit covariance for the secondary user is given as
\vspace{-5pt}\begin{align*}
\mv{S}_x^\ast=\left[\begin{array}{rr}5.7228 & 1.4217 \\ 1.4217  & 4.2772 \end{array} \right].
\end{align*}With this transmit covariance, the capacity of the secondary user is $6.7893$ bps/Hz.  It can be shown that the eigenvectors of $\mv{S}_x^\ast$ are different from the right-singular vectors of $\tilde{\mv{G}}$. As a result, the channel SVD based linear precoding/decoding design given in (\ref{optimal transmit signal}) and (\ref{single-user signal eigenmode}) cannot be applied to diagonalize the secondary MIMO channel and achieve the capacity.

With Theorem \ref{theorem1}, a new pair of linear precoder and decoder for the secondary link is obtained as
\vspace{-3pt}\begin{align*}
\mv{V}=\left[\begin{array}{rr}-2.3544 & \hspace{-5pt}0.4238 \\ -0.9358  & \hspace{-5pt}-1.8443 \end{array} \right],  ~ \mv{U}=\left[\begin{array}{rr}-0.9870 & \hspace{-5pt}0.1607 \\ 0.1607  & \hspace{-5pt}0.9870 \end{array} \right].
\end{align*}With this new design, it can be shown that $\mv{V}\mv{V}^H=\mv{S}_x^\ast$. Moreover, $\mv{U}^H\tilde{\mv{H}}\mv{V}$ is a diagonal matrix, and the achievable rate of the secondary user is $6.7893$ bps/Hz, which is the same as its channel capacity.

\vspace{-5pt}
\section{Conclusion}\label{sec:Conclusion}
This letter studied the practical ``channel-diagonalizing'' linear precoding/decoding design to achieve the capacity of the point-to-point MIMO channel given an arbitrary transmit covariance when the channel is perfectly known at the transmitter and receiver. We proposed a closed-form solution for this problem and verified its effectiveness in various multi-user MIMO systems. This result was shown to be particularly useful for diagonalizing the MIMO channel and yet achieving the capacity in multi-user MIMO systems when the optimal transmit covariance of each user is not channel SVD based as in the conventional single-user MIMO.

\end{document}